\documentclass[letterpaper,UKenglish,cleveref, autoref, thm-restate,nolineno]{socg-lipics-v2021}

\makeatletter
\newcommand{\@hideLIPIcs}{}
\makeatother

\usepackage{todonotes}
\usepackage{cleveref}
\usepackage{graphicx}
\usepackage{complexity}
\usepackage{xspace}
\usepackage{caption}

\graphicspath{{figures/}}

\newcommand{\defn}[1]{\textbf{#1}}
\renewcommand{\R}{\mathbb{R}}
\newcommand{\Z}{\mathbb{Z}}

\renewcommand{\CH}[1]{\text{CH}(#1)}
\newcommand{\Int}[1]{\text{Int}(#1)}
\newcommand{\remove}[1]{}
\newcommand{\wg}{\widetilde{g}}

\newcommand{\cF}{\mathcal{F}}
\newcommand{\cI}{\mathcal{I}}
\newcommand{\cone}[1]{\text{Cone}(#1)}

\newcommand{\AG}{\textsc{ArtGallery}\xspace}
\newcommand{\BAG}{\textsc{BoundaryArtGallery}\xspace}
\newcommand{\CAG}{\textsc{ContiguousArtGallery}\xspace}
\newcommand{\PS}{\textsc{PolygonSeparation}\xspace}
\newcommand{\PointS}{\textsc{PointSeparation}\xspace}
\newcommand{\SegS}{\textsc{SegmentSeparation}\xspace}
\newcommand{\SegPolyS}{\textsc{SegmentPolygonConvexSeparation}\xspace}
\newcommand{\SegCPS}{\textsc{SegmentConvexPolygonSeparation}\xspace}
\newcommand{\AAC}{\textsc{AnalyticArcCover}\xspace}

\title{The Analytic Arc Cover Problem and its Applications to Contiguous Art Gallery, Polygon Separation, and Shape Carving} %

\titlerunning{The Analytic Arc Cover Problem} %

\author{Eliot W. Robson}{
Department of Computer Science,
University of Illinois Urbana-Champaign
}{
erobson2@illinois.edu
}{%
https://orcid.org/0000-0002-1476-6715%
}{%
}
\author{Jack Spalding-Jamieson}{Independent}{jacksj@uwaterloo.ca}{https://orcid.org/0000-0002-1209-4345}{}
\author{Da Wei Zheng}{
Department of Computer Science,
University of Illinois Urbana-Champaign
}{
dwzheng2@illinois.edu
}{%
https://orcid.org/0000-0002-0844-9457%
}{%
Research supported in part by an NSERC PGSD.
}

\authorrunning{E. Robson, J. Spalding-Jamieson, D.W. Zheng}

\Copyright{Eliot Robson, Jack Spalding-Jamieson, Da Wei Zheng}



\category{} %

\relatedversion{} %

\acknowledgements{Thanks to Thomas C. Shermer for proposing the contiguous art gallery problem at the open problem session at the
Canadian Conference on Computational Geometry 2024 (CCCG24) at Brock University, Canada, and for fruitful discussions at the conference.}%

\begin{document}

\maketitle

\begin{abstract}
We show the following problems are in \P:
\begin{enumerate}
    \item The contiguous art gallery problem --  a variation of the art gallery problem where each guard can protect a contiguous interval along the boundary of a simple polygon. This was posed at the open problem session at CCCG '24 by Thomas C. Shermer.
    \item The polygon separation problem for line segments -- 
    For two sets of line segments $S_1$ and $S_2$, find a minimum-vertex convex polygon
    $P$ that completely contains $S_1$ and does not contain or cross any segment of $S_2$.
    \item Minimizing the number of half-plane cuts to carve a 3D polytope.
\end{enumerate}
To accomplish this, we study the \emph{analytic arc cover problem} -- an interval set cover problem over the unit circle with infinitely many implicitly-defined arcs, given by a function.
\end{abstract}

\section{Introduction}
\label{sec:introduction}

Many classic problems in computational geometry are minimum covering problems.
One class of examples are
art gallery problems~\cite{ORourke87,Shermer92,Urrutia00,ORourke04}
which asks for the a minimum number star-shaped polygons that cover a given polygon.
Each star-shaped polygon describes a region that can be seen by a single guard. 
Some variants of art gallery allow guards to be placed anywhere,
while others only allow guards to be placed at a finite set of points (such as vertices of a polygon).
Many variants of the former type are $\exists\R$-complete, while variants of the latter type are almost universally in \NP,
and are often in \P.
This distinction is not surprising: 
When there are infinitely many choices for the covering sets 
it is often not clear if the problem is even in \NP.

In this work, we study three problems
with no immediate proofs that they are in any complexity class smaller than $\exists\R$.
The first problem is a variant of the art gallery problem with a restriction that each guard can only be responsible for a contigious region.
The other two problems relate to separation.
We will show that each of these problems can be reduced to a problem
call the \emph{analytic arc cover} problem.
We establish machinery for dealing with some ``well-behaved'' infinite classes of possible covering sets,
allowing us to show all three of our problems are in \P.

\subsection{Art Gallery Problem}

\begin{figure}
\centering
\includegraphics[scale=0.3,page=2]{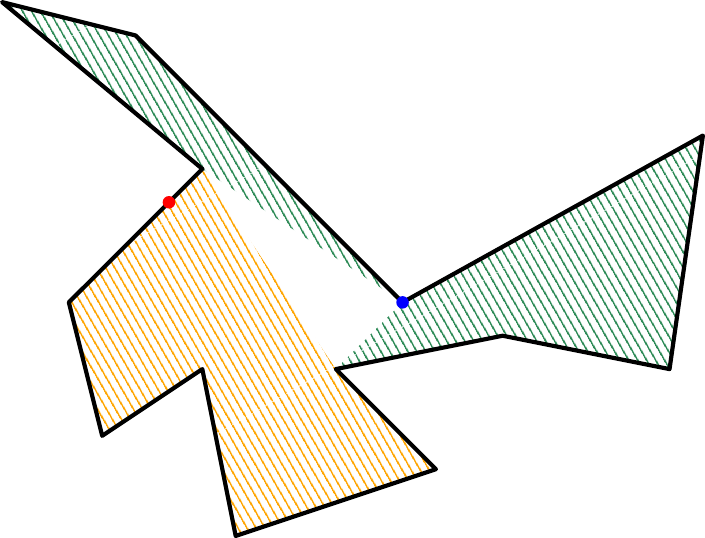}
\hspace{4em}
\includegraphics[scale=0.3,page=4]{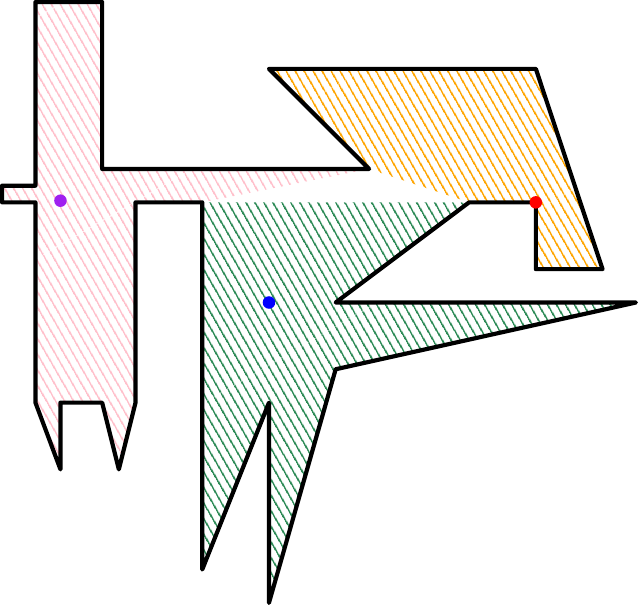}
\caption{Some examples of \CAG, with optimal solutions.
On the right, it is necessary to place a guard at a non-vertex point.
}
\label{fig:art-gallery-examples}
\end{figure}

Given a simple polygon $P$ and points $x,y \in P$, a guard standing at $x$ sees the point $y$ if the line segment $xy$ is contained in $P$. 
A set $S$ of points is said to \emph{guard} the polygon $P$ if every point of $P$ can be seen by some guard. Minimizing the cardinality of $S$ is known as the \AG problem\footnote{Throughout this paper, we will abuse notation and refer to problems like \AG as both the optimization variant to minimize the set of guards, as well as the decision variant, where we wish to decide if $k$ guards are sufficient.}, and is a very well-studied problem in computational geometry (See one of the numerous surveys dedicated to the problem~\cite{ORourke87,Shermer92,Urrutia00,ORourke04}).
This problem was first formulated by Victor Klee in 1976 (see O'Rourke \cite{ORourke87}). Many variations have been studied since then, such as when the guards are restricted to the boundary or the vertices of the polygon. There are other variants like the \BAG problem where guards are allowed to be anywhere but only the boundary of $P$ needs to be guarded.

All of the aforementioned variants were shown to be \NP-hard \cite{Aggarwal84,LeeL86,Laurentini99}. 
It can be seen that variants that allow for guards in the interior may not even be in \NP, since it was shown that an optimal solution to \AG may need guards at irrational points~\cite{AbrahamsenAM17}.
Further evidence towards this was recently given by 
Abrahamsen, Adamaszek, and Miltzow~\cite{AbrahamsenAM22},
who showed that Victor Klee's original formulation is $\exists\R$-complete.
Subsequent work by Stade~\cite{Stade22} shows that many of the variants of art gallery are also $\exists\R$-complete.

One salient feature of many of the hardness proofs for \AG and its variations is that a single guard may be responsible for disjoint regions of the polygon, even with respect to guarding just the boundary of $P$. In contrast, we can consider the \CAG problem -- a variation of the \BAG problem where each guard can only be responsible for a contiguous section of the boundary, see \Cref{fig:art-gallery-examples}. This problem was described by Thomas C. Shermer at the CCCG '24 open problem session, where he asked the following question:

\begin{center}
\textit{Is the guarding of disjoint regions necessary for the hardness proofs of \AG and variations like \BAG?\\
What is the complexity of \CAG? Is it even in \NP?}
\end{center}

\subsection{Minimum Polygon Separation}

Given a convex polygon $P$ that is contained inside of another convex polygon $Q$, the minimum polygon separation problem  asks for the polygon $S$ with the minimum number of vertices such that $S$ contains $P$ and is contained within $Q$. We denote this problem as \PS.
\PS was first studied by Aggarwal, Booth, O'Rourke, Suri, and Yap~\cite{YapABO89},
who also showed that any minimal $S$ is convex.
This problem is related to the problem of finding a convex polygon with the minimum number of sides that separates two sets of points $S_1$ with $S_2$ that was studied by Edelsbrunner and Preparata~\cite{EdelsbrunnerP88}. We denote this problem as \PointS.

One feature of the algorithm of \cite{YapABO89} is that they work in the real RAM model of computation~\cite{Shamos78}, where it is assumed that elementary operations (like $+$, $-$, $\times$, and $\div$) can be performed in $O(1)$ time on real numbers.
However, the algorithm they present involves iteratively composing many functions together, and ultimately solving an equation. 
While \cite{YapABO89} claim that their algorithm runs in $O(n\log n)$ time when the input consists of $n$ points, it is unclear what the runtime of the algorithm is in other models of computation (e.g. in a Turing Machine model where inputs are points with rational coordinates). This is because their algorithm may compose $O(n)$ different functions, leading to a function with very high bit complexity, while \cite{YapABO89} assume an equation with such a function can be solved in $O(1)$ time.

One can also consider another similar problem:
The problem of separating two collections of line segments by a convex polygon with the minimum number of sides.
We call this problem \SegS.
It can be seen that
\SegS
generalizes the problem of \PS when we wish to separate the edges of outer polygon with
the inner polygon.
Furthermore, the \PointS problem is a special case of \SegS when all segments are of zero length.
To the best of our knowledge, \SegS has not been studied before, and the algorithms of \cite{YapABO89} and \cite{EdelsbrunnerP88} do not immediately work for this problem.

Moreover, while the problems \PS and \PointS were implicitly shown to be in \P\xspace using the algorithms of \cite{YapABO89} and \cite{EdelsbrunnerP88} respectively (with a larger time complexity in the Turing Machine model of computation rather than the real RAM models used in those papers),
the same techniques do not immediately apply to \SegS, 
so it is not immediately clear that \SegS is even in \NP.

\subsection{Polytope Carving}

\PS is related to the problem of \emph{carving} one shape out of another.
One could also carve a shape out of an arbitrary containing shape,
which is equivalent to asking for a set of cuts separating the shape from the ambient space.
Two-dimensional carving was first studied by Overmars and Welzl \cite{OvermarsW85}, where they aimed to find the cheapest sequence of line cuts to carve out a convex polygon out of a piece of flat material. 
It has also been studied in the context of
rays cuts~\cite{DaescuL06, Tan05} and line segments cuts~\cite{DemaineDK01, DumitrescuH13}. 

Three dimensional versions of carving have also been studied in the form of plane cuts \cite{AhmedHI11}, line cuts \cite{JaromczykK03}, half-plane cuts \cite{RobsonSJZ2024}, and ``sweeping'' ray cuts~\cite{RobsonSJZ2024}.
The work of Robson, Spalding-Jamieson and Zheng~\cite{RobsonSJZ2024} on half-plane cuts left open the question of computing the minimum number of half-plane cuts needed to carve out a specified 3D polytope from the ambient space of $\R^3$.

\subsection{Our Contribution}
We answer the question posed by Thomas C. Sherman with the following theorem, showing that adding contiguity constraints for \BAG makes the problem significantly simpler.
We also show that \SegS is in \P.
\begin{theorem}
\label{thm:P-cag}
    \CAG is in \P.
\end{theorem}

\begin{theorem}
\label{thm:P-seg-s}
    \SegS is in \P.
\end{theorem}

We show that a question about minimizing the number of half-plane cuts to carve a 3D polytope, reduces to multiple instances of \SegS. By the theorem above, we can conclude that finding the minimum cuts to carve a polytope is also in \P.

\begin{theorem}
\label{thm:P-3d}
Minimizing the number of half-plane cuts to carve a 3D polytope
is in \P.
\end{theorem}

To prove our results, we reduce both \SegS and \CAG to a problem we call \AAC.
This problem is a version of the interval cover problem on a circle with infinite number of intervals. Each interval is described by the counterclockwise segment between two points $a$ and $b$ on the circle $S^1$. We denote this as an half-open interval $[a, b)$, containing $a$, and not containing $b$, and call this an \defn{arc}. This infinite set of intervals is given implicitly as a function, as can be seen in the following definition.
\begin{definition}
Let $g:S^1\to S^1$ be a function that maps points on the unit circle $S^1$ to other points on the unit circle $S^1$. 
\AAC asks: given $g$, find the minimum set $X\subset S^1$ such that the set of counter-clockwise arcs $\{[x,g(x)):x\in X\}$
covers $S^1$.
\end{definition}

We show that \CAG reduces to \AAC with piecewise linear-rational functions, and \SegS reduces to \AAC with a two dimensional analogue. We crucially use the fact that a composition of two linear rational functions yields another linear rational function.

\subsection{Concurrent work}

Two other groups also concurrently investigated the \CAG problem.

Merrild, Rysgaard, Schou, and Svenning~\cite{MerrildRRS24} give a polynomial time algorithm for the \CAG problem in the real RAM model of computation, but do not bound the bit complexity of the intermediate numbers produced by the algorithm. They posed the question of membership in \P\ as an open problem, which our results address.
They also posed the question of an algorithm for polygons with holes, to which our methods extend (see \cref{rmk:holes}).

Biniaz, Maheshwari, Mitchell, Odak, Polishchuk, and Shermer~\cite{biniaz2024contiguousboundaryguarding}
also provide a polynomial time algorithm for the \CAG problem.
In particular, like us, they were able to provide one running on a Turing machine, implying membership in \P.
Their approach is more combinatorial than ours, although aspects of their proof bear similarities to ours.

In our work, we give a framework for solving a larger class of problems in polynomial time on a Turing machine (i.e., membership in \P), including the \CAG problem.

Since the initial preprint of our paper, all three works (including ours) have been presented as a combined paper~\cite{combinedsocgpaper}, which also led to some improved insights among the three sets of results.

\subsection{Organization}

The remainder of this paper is organized as follows.
In \Cref{sec:arc-cover}, we formally define and then solve the \AAC problem.
In \Cref{sec:art-gallery}, we reduce the \CAG problem to the \AAC problem, proving \Cref{thm:P-cag}.
Similarly, in \Cref{sec:sep}, we reduce the \SegS problem to the \AAC problem, proving \Cref{thm:P-seg-s}. 
Finally, in \Cref{sec:carving}, we prove \Cref{thm:P-3d} via a reduction to \SegS, leveraging \Cref{thm:P-seg-s}.

\section{The Analytic Arc Cover Problem}
\label{sec:arc-cover}

Throughout the paper we will work with $S_1$ the unit circle centered at the origin. For points $a,b$ in $S_1$, we will use the notation $[a,b)$ to denote the half-open counter-clockwise arc from $a$ to $b$ that includes $a$ but does not include $b$ (the choice of orientation is arbitrary).

\subsection{Warmup}

We begin with a definition of the interval cover and arc cover problems.
\begin{definition}
Let $\mathcal I$ be a set of $n$ half-open intervals of the form $[a, b)$ covering the unit interval $[0,1)$.
The \textsc{IntervalCover} problem asks for the minimum subset of intervals $I\subset\mathcal I$
that can cover $[0, 1)$.
\end{definition}
\begin{definition}
Let $\mathcal I$ be a set of $n$ half-open arcs covering the unit circle $S^1$.
The \textsc{ArcCover} problem asks for the minimum subset of arcs $I\subset\mathcal I$
that can cover $S^1$.
\end{definition}

While the set cover problem is NP-Hard~\cite{GareyJ79}, the interval set cover problem and the arc cover problem are each solvable in $O(n\log n)$ time~\cite{clrs, hsu1991linear}.
In particular, there is a folklore greedy algorithm for \textsc{IntervalCover} that repeatedly looks for the leftmost uncovered point $p$, and finds the interval $[a,b)$ that covers $p$ and has the maximum value of $b$.
This same algorithm can be applied to \textsc{ArcCover} if one fixes a branch cut (i.e., starting point),
although it only achieves an additive approximation:
\begin{lemma}
Let $\mathcal I$ be an instance of \textsc{ArcCover} where the minimum arc cover has size $k$. Fixing a point $p\in S^1$ and running the greedy algorithm counter-clockwise from $p$ yields a solution with at most $k+1$ intervals.
\end{lemma}
\begin{proof}
Let $I_1, \dots, I_k$ be an optimal solution to the arc cover problem with $p\in I_1$. Let $I_1', \dots, I_\ell'$ be the solution found by the greedy algorithm starting at $p$.
Observe that by the greedy algorithm has the counter-clockwise point of $I_j'$ farther than the counter-clockwise point of $I_j$.
Taking the interval $I_1$ along with $I_1',\dots I_k'$ would cover all of $S^1$. Thus $\ell \le k+1$.
\end{proof}

\subsection{Formulating the Analytic Arc Cover Problem}

In this section we will generalize \textsc{ArcCover} to carefully chosen settings where the set of arcs $\cI$ can be infinite.
Given a set of (possibly infinte) arcs $\mathcal I$, we can define the \defn{next-generator} $g:S^1\to S^1$ that takes as input a point $t\in S^1$ computes the ``farthest'' interval that covers $t$ with respect to the branch cut at $a$:
\[ g(t) = \sup\{ b \mid [a,b)\in \mathcal I \text{ and } t\in [a, b)\}\]
For simplicity and to avoid pathological examples, in this paper we will always assume that $\mathcal I$ has the property that for every $t$, there is some interval that contains $[t, g(t))$.
That is, the $\text{sup}$ can be replaced with a $\text{max}$.
This can be thought of as analogous to a completeness property in analysis.

The core problem we address is what we call the \AAC problem. In its most general form, the problem is as follows:
\begin{definition}
Let $\mathcal I$ be a set of (possibly infinite) half-open arcs covering the unit circle $S^1$ with a next-generator $g$.
The \AAC problem asks for the minimum integer $k$ for which there exists an $x$ where:
\[[x, g(x))\cup [g(x), g(g(x)))\cup \cdots \cup [g^{(k-1)}(x),g^{(k)}(x)) = S^1.\]
\end{definition}

This definition coincides with solving \textsc{ArcCover} using the (possibly infinite) set of intervals $\cI$.
To actually discuss algorithms for this problem, it will be helpful to work in a particular space homeomorphic to $\mathbb{R}$,
rather than $S^1$.
Our space of choice is $\Z \times S^1$,
equipped with a canonical branch cut (i.e. the positive $x$ axis).
This space has a natural ordering, lexicographical with counterclockwise ordering relative to the branch cut.
A next-generator function $g:S^1\to S^1$
can be \defn{lifted}
to a function $\wg:\Z \times S^1\to \Z \times S^1$
so that
$\wg(z,x)=\left(z+1,g(x)\right)$ if $[x,g(x))$ contains the branch cut point,
and
$\wg(z,x)=\left(z,g(x)\right)$ otherwise.
We call $\wg$ the \defn{lifted next-generator}.
See \cref{fig:lifted-functions} for a visualization.

\begin{figure}[ht]
\centering
\includegraphics[scale=0.8,page=1]{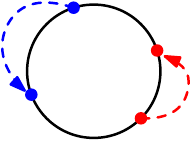}
\hspace{3em}
\includegraphics[scale=0.6,page=6]{lifted-functions}
\caption{Two mappings of a next-generator and the same mappings in the lifted next-generator.}
\label{fig:lifted-functions}
\end{figure}

Observe a next-generator $g$ has the property that for any $t_1, t_2\in S^1$ where $t_2\in [t_1, g(t_1))$, then $g(t_2) \not \in [t_1, g(t_1))$.
The lifted next-generator $\wg$ is in fact monotonic with our choice of lexigraphical ordering on $\Z\times S^1$.
Furthermore, the lifted next-generator has the feature that $\wg(z, x) \le (z+1, x)$. 
We call functions from $\Z\times S^1 \to Z\times S^1$ that have this property \defn{proper} functions.

In our applications, rather than directly working with the infinite set $\mathcal{I}$, we will be able to compute a compact representation of the next-generator,
which will also induce a straightforward representation of the lifted next-generator.
The possible forms of the next-generator function can be quite varied.
Typically, one would want an oracle to evaluate the next-generator function,
although in our case that would only allow us to get an additive $+1$ approximation to the optimal solution.
Instead, we will require that our lifted next-generator function
comes from a family of functions $\cF$ with the following properties:
\begin{enumerate}
    \item $\mathcal F$ is closed under composition.
    \item One can perform an \defn{existential threshold test}: For a function $f\in\mathcal F$, one can check if there exists $x\in S^1$ such that $f(0,x)\geq (1,x)$ (as points on $\Z\times S^1$), and output a representation of that point as a witness. In other words, we can test if $f$ is a proper function.
\end{enumerate}
We will call such a function family \defn{strongly computable},
although we note that this is not a generalization of a function family with an evaluation oracle.

\begin{theorem}
\label{thm:slow-alg}
Let $\cF$ be a strongly computable function family,
and let $\wg\in\cF$ be a lifted next-generator for some set of intervals $\cI$.
Let $I^*$ be a minimum cover of $\mathcal I$ of size $k$.
Then \AAC on $\cI$ can be solved in $k-1$ function compositions and $k$ existential threshold tests.
\end{theorem}

\begin{proof}
Let $g^{(i)}$ denote the function $\wg$ composed with itself $i$ times, so $g^{(i)} = g\circ g^{(i-1)}$.
Thus we can compute $g^{(i)}$ for each $i$ via function compositions
For each $g^{(i)}$ we can compute an existential threshold test to determine if there is a solution of size $i$,
and halt if so.
This algorithm terminates after $k-1$ function compositions and $k$ existential threshold tests.
\end{proof}

Importantly, note that to use this algorithm, one only needs to provide a representation of $\wg$. In the next section, we will show that if $g$ is a piecewise linear rational function,
a representation of $\wg$ follows,
and the above algorithm runs in polynomial time.

\begin{remark}
It is possible to reduce the number of function compositions and existential threshold tests in \cref{thm:slow-alg} from $k$ down to $O(\log k)$ by using repeated doubling to binary search for $k$. We omit such details as we do not focus on exact runtimes in this paper.
\end{remark}

\subsection{Representations of the Unit Circle}

Since we are dealing with covering problems on the unit circle $S^1$, we need to have useful ways of representing points along it.
There is one obvious method:
Represent points by their angle.
Rather than using units of radians or degrees,
it suffices to simply represent each angle as an element of the range $[0,1)$.
We will use this in \cref{sec:art-gallery}.
We refer to this as the \defn{unit-interval representation}.

The other method we will use in \cref{sec:sep} is more involved.
We start by describing a slightly simpler method:
Represent each point in $S^1$ by the coordinates of its embedding in $\R^2$.
This is equivalent to using the two-dimensional real-valued unit vectors
as representatives.
Although we could use these for \cref{sec:sep},
some of the details would be more involved,
so it will make things simpler later to add a bit of extra complexity to the representation.
Specifically,
we represent each point $x$ in $S^1$
with \emph{any} of the points in $\R^2\setminus\{(0,0)\}$
along the ray from the origin through $x$ in the embedding of $S^1$ in $\R^2$.
This representation will greatly aid with the presentation of some computational purposes,
since it will allow us to bypass directly representing angles in \cref{sec:sep}.
It should be noted that some care will have to be taken to ensure the next-generator functions are well-defined in this form.
We refer to this as the \defn{ray representation},
although computationally we will store the rays simply as a point along the ray (since the origin of all rays are identical).
The point not necessarily of unit distance from the origin, but rather the rays induce an equivalence relation.

In both cases, the primary computational aspect we need is that any two points
along $S^1$ can be compared in their counter-clockwise order relative to a branch cut
(note that a branch cut itself is a point in $S^1$).
This is trivial in the unit-interval representation.
In the ray representation,
this can be done with a cross product.
Importantly, this means that obtaining a lifted next-generator
is as simple as performing such a comparison with some canonical branch cut.

\subsection{Piecewise Linear Rational Functions}

A \defn{$d$-dimensional linear rational function} is a function of the form:
\[ \overline{x}\mapsto \frac{A\overline{x}+\overline{b}}{C\overline{x}+\overline{d}} \qquad 
\text{ for $A,C\in\Z^{d\times d}$ and $\overline{b},\overline{d}\in\Z^d$ (division is pointwise).} \]
We will only consider $d\in \{1, 2\}$.
As discussed in the previous section, the $2$-dimensional representation specifically corresponds to
the ray representation,
in which case we require that $\overline{b}=\overline{d}=\overline{0}$.
The reason for this is simple: this is enough to guarantee that all points along a ray map to points along the same ``next'' ray.

We will consider \defn{piecewise linear rational functions} that are linear rational functions when the domain is partitioned into several intervals (i.e. arcs along the unit circle).
We call the end points of these intervals \defn{boundary points},
which themselves are just points in $S^1$ in our chosen representation.
Note that in the ray representation,
boundary points are entire rays
(represented by a point along the ray).
Since we have the ability to compare ordering relative to a branch cut, there is no fundamental difference
between the unit-interval and ray representations.

The main reason for the use of linear rational functions
is that they are closed under composition.
This is also true of piecewise linear rational functions.

In this paper, we will use piecewise linear rational functions with rational coefficients
of one or two variables in the aforementioned domains,
whose boundary points (unlike the coefficients) use values from the field
of \defn{first order radical numbers}.
That is, numbers of the form $a+\sqrt{c}$,
where $a,b,c$ are rationals.
Surprisingly, this will not cause any significant difficulties in our constructions, as we will not have higher order radicals.
Note that the sum of square roots is a major bottleneck in proving NP-completeness of geometric problems in the Turing machine model of computation \cite{o1981advanced,topp_p33}. In our case,
we will only deal with two square roots at a time (at most one on each side of a comparison), so we avoid this problem.

We defer the details of the following theorem to \cref{subsec:aac_deets}, where we prove that the algorithm given in 
\cref{thm:slow-alg} runs in polynomial time:

\begin{theorem}
\label{thm:composition_lrf}
If $g$ is a piecewise linear rational function
for either the unit-interval or ray representations of $S^1$, 
the \AAC problem can be solved in time polynomial in the size of the optimal solution $k$, 
the combined bit-complexity of the end points of the pieces and each linear rational function
(i.e. polynomial in $k$ and in the bit complexity of the representation of $g$).
\end{theorem}

\subsection{Details for piecewise linear rational functions}
\label{subsec:aac_deets}
We care about the \defn{bit complexities} of the rational functions whose end points are radicals. 
We define bit complexity as follows.
\begin{definition}
The bit complexity of:
\begin{enumerate}
    \item an integer $x$ is the number of bits that it takes to store $x$, i.e. the bit complexity is $\lceil \log_2 x\rceil$.
    \item a linear rational function $f$ is the maximum bit complexity of its coefficients.
    \item a rational number is the maximum of the bit complexity of the denominator or numerator.
    \item a radical $a+b\sqrt{c}$ is the maximum bit complexity of $a$, $b$, and $c$. 
\end{enumerate}
\end{definition}
Note that these definitions ensure that the bit complexities of these objects are proportional to the number of bits needed to store them on a Turing machine. Furthermore, they have the following properties whose proofs are straightforward and thus omitted:
\begin{lemma}
\label{lemma:bit_complexity}
The following are facts about bit complexities.
\begin{enumerate}
    \item The bit complexity of $f\circ g$,
    the composition of two linear rational functions $f$ and $g$,
    is at most the sum of the bit complexities of $f$ and $g$.
    \item The bit complexities of $x\cdot y$,
    the product of two rational numbers $x$ and $y$,
    is at most the sum of the bit complexities of $x$ and $y$
    \item The bit complexities of $f(x)$,
    the evaluation of a linear rational function $f$ 
    on a first-order radical $x$,
    is at most the sum of complexity of $f$ and $x$.
\end{enumerate}
\end{lemma}

\begin{lemma}
\label{lemma:lr-max-min}
Let $f_A$ and $f_B$ be two piecewise linear rational functions
in either the unit-interval or ray representation,
with domains $D_A,D_B\subset S^1$ (none of $D_A,D_B,S^1$ are necessarily equal).
Then, for any choice of branch cut, the following are also piecewise linear rational functions with domain $D_A\cup D_B$:
\begin{itemize}
    \item $f_{\max}(x) = \begin{cases}
        \max(f_A(x), f_B(x)) & x\in D_A\cap D_B\\
        f_A(x) & x\in D_A\setminus D_B\\
        f_B(x) & x\in D_B\setminus D_A
    \end{cases}$
    \item $f_{\min}(x) = \begin{cases}
        \min(f_A(x), f_B(x)) & x\in D_A\cap D_B\\
        f_A(x) & x\in D_A\setminus D_B\\
        f_B(x) & x\in D_B\setminus D_A
    \end{cases}$
\end{itemize}
Moreover, if $f_A$ and $f_B$ respectively have $n$ and $m$ pieces within their respective domains,
then $f_{\max}$ and $f_{\min}$ each have $O(n+m)$ boundary points,
and the description complexity of the resulting function is at most the sum of the description complexity of $f_A$ and $f_B$.
\end{lemma}
\begin{proof}
Consider $f_{\max}$ (the proof is identical for $f_{\min}$).
Split into maximal arcs of $S^1$
contained entirely within some piece of $f_A$ and some piece of $f_B$, or entirely out of the domain of one but not the other.
There are at most $O(n+m)$ such maximal arcs.
Within each arc entirely out of the domain of one,
there are no additional boundary points added.
To complete the proof, it suffices to show that within each shared maximal arc,
there are $O(1)$ points $x$ such that $f_A(x)=f_B(x)$, and all the additional boundary points are first order radical numbers,
since that would imply $O(n+m)$ total boundary points.

In the unit-interval representation,
the equation $f_A(x)=f_B(x)$
is asking for equality of general one-dimensional two rational functions,
which is a quadratic equation,
with at most two solution (unless the two are equal everywhere, which does not induce new boundary points).

In the ray representation,
the equation
$f_A(x)=f_B(x)$
is equivalent to requiring that
$\frac{A\overline{x}}{B\overline{x}}=\frac{C\overline{x}}{D\overline{x}}$,
for matrices $A,B,C,D\in\R^{2\times2}$,
or
$A\overline{x}\odot D\overline{x}=C\overline{x}\odot B\overline{x}$,
where $\odot$ denotes point-wise multiplication.
This is a system of (two) quadratic equations in two variables with pure quadratic terms
(all terms are of the form $ax_i^2$ for one of the variables $x_i$),
which can be reduced to a system of general quadratic equations in one variable,
and hence has at most $2$ solutions,
both of which are first-order radical numbers.
\end{proof}

\begin{lemma}
\label{lemma:comp-lifted}
Let $\wg$ be a lifted next-generator function
such that $g$ is a piecewise linear rational function (in either the unit-interval or ray representations)
with at most $n$ boundary points,
each of bit complexity $\leq b$.
Then $\wg^{(m)}$ has at most $\text{poly}(n,m)$ boundary points within one period,
i.e. within $\left(S^1\times\{0\}\right)\subset\left(S^1\times\Z\right)$.
Moreover, each of them is of bit complexity also at most $\text{poly}(b,m)$.
\end{lemma}

\begin{proof}
First, note that each $\wg^{(m)}$ is monotonic and periodic with period one as well.
We will now prove the first fact inductively.
For the base case, $\wg^{(1)}$ has exactly $n$ boundary points in $\left(S^1\times\{0\}\right)$.
We claim that $\wg^{(m)}$ has at most $nm$ boundary points in $\left(S^1\times\{0\}\right)$.
Assume this is true for $m\geq1$.
The boundary points of $\wg^{(m+1)}=\wg\circ\wg^{(m)}$
are the boundary points of $\wg$ in addition to the boundary points of the form $\wg^{(-1)}$(x),
where $x$ is a boundary point of $\wg^{(m)}$.
Note that $\wg^{(-1)}$
is also a periodic and monotonic piecewise linear rational function (with the same period),
so there are $\leq nm$ points of this form within each period.
Thus $\wg^{(m+1)}$ has at most $(m+1)n$ boundary points within each period.

The bit complexity of these boundary points follows a similar argument.
The base case is obvious.
Assume for induction that the boundary points of $\wg^{(m)}$ have bit complexity at most $mb$.
The bit complexity of $\wg^{(-1)}$ is at most $b$
in both representations.
Hence, the bit complexity of each additional boundary point added
is at most $(m+1)b$ by \cref{lemma:bit_complexity}
\end{proof}

Observe that the above lemma directly implies \cref{thm:composition_lrf}.
In particular, an existential threshold test can be performed
for a lifted piecewise linear rational function $\wg$ in polynomial time,
by evaluating $\wg$ at each boundary point (within one period),
and at all the zeroes (within one period) of the derivatives of the pieces of $x\mapsto\wg(x)-x$
(by Fermat's stationary point theorem),
so the result follows from 
\cref{thm:slow-alg}.

\section{The Contiguous Art Gallery Problem}
\label{sec:art-gallery}
As a reminder, we can formalize \CAG as follows:
\begin{definition}
Let $P$ be a simple polygon.
The \CAG problem
asks to find the minimum set
of contiguous paths covering the boundary,
such that each path is visible in its entirety by some point inside the polygon.
\end{definition}

The intuition for the reduction to \AAC is quite natural:
The circle is topologically equivalent to the polygon,
and so the arcs $\mathcal I$ will represent (maximal) contiguous paths along the polygon
that have a non-empty intersected visibility region.
Defining the next-generator function $g$ is similarly straightforward.
The main focus of this section is to obtain a piecewise linear rational representation of $g$.

We present the proof as a series of lemmas.
The sequence of lemmas can be viewed a series of progressively more general carefully chosen toy examples,
each building on the previous by making use of \cref{lemma:lr-max-min}.
The first few lemmas will be mostly standalone:

\begin{figure}[ht]
\centering
\includegraphics[scale=0.6,page=1]{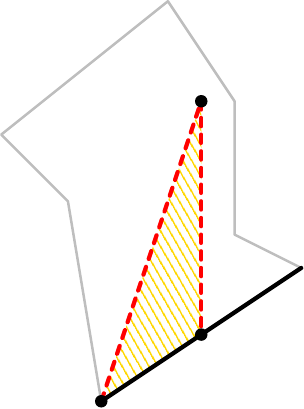}
\caption{Two two points along a shared edge both visible from a shared point, and the corresponding (empty) triangular region.}
\label{fig:suffix-prefix-unnecessary}
\end{figure}

\begin{lemma}
\label{lemma:suffix-prefix-unnecessary}
Let $p,p'$ be two points along the same edge $e$ in a simple polygon $P$.
Then if $y$ is a point in $P$ from which both $p$ and $p'$ are visible,
all points along the edge $e$ in between $p$ and $p'$ are also visible from $y$.
\end{lemma}

\begin{proof}
The $P$ is a simple polygon, the triangle with vertices $p,p',y$ must be contained entirely inside the polygon,
including each of its edges.
See \cref{fig:suffix-prefix-unnecessary}.
\end{proof}

\begin{lemma}
\label{lemma:greedy-choice}
Let $P$ be a polygon and $x$ be a point along the boundary.
Then, it is possible to compute, in polynomial time,
the furthest edge around the polygon counter-clockwise
so that some guard point inside the polygon
has a visibility region
that includes the entire boundary
counter-clockwise from $x$ to the edge,
as well as some part of the edge beyond the endpoint.
\end{lemma}
Note that this is slightly different from computing the furthest point counter-clockwise,
although it does narrow down which edge that point is on.

\begin{proof}
Starting at $x$,
traverse the polygon in counter-clockwise order.
For each edge encountered,
start by computing the visibility region of the shared vertex in the path so far.
Intersect this with the visibility region computed so far.
By \cref{lemma:suffix-prefix-unnecessary}, this is no different than intersecting the visibility polygons~\cite{ga-lacvpp-81}
over the continually changing point of the boundary for each edge,
which can be done in polynomial time~\cite{wa-hsrupas-77}.
Then,
if the visibility region becomes empty, halt,
since only a partial edge is visible.
Before the next iteration,
check if any part of the visibility region is contained in the half-plane
aligned with the edge,
directed inside the polygon.
If not, halt,
since no part of the next edge will be visible.
\end{proof}

The remaining lemmas build on each other.
Each deals with a pair of line segments $e$ and $f$,
and we will use some common notation for brevity.
In particular, each lemma will map points on $e$ to a particular point on $f$,
each in parametrized form.
These represent edges in the general problem, although we need not consider the general problem until the main proof.
Specifically, assume
$e=(p_e,p_e')$ and $f=(p_f,p_f')$.
Then a point along each can be defined by parameters $t$ and $r$,
as $\left(p_e+t(p_e'-p_e)\right)$ and $\left(p_f+r(p_f'-p_f)\right)$.
For concise statements, we let $e_t$ and $f_r$ denote each of these respective parametrized points.

\begin{definition}
    For points $x,y,z$, $\cone{x,y,z}$ is the counter-clockwise filled region bounded by the two rays from $y$ through $x$ and $z$.
\end{definition}

\begin{lemma}
\label{lemma:restricted-toy}
Let $e,f$ and $l$ be non-crossing line segments in $\R^2$,
with $e,f$ as before.
Assume that neither $e$ nor $f$ is crossed by the extension of $l$ to a line.
Let $x_1,x_2$ be points in $\R^2$.
Assume $l$, $x_1$, and $x_2$
all lie within the intersection of two half-planes
whose boundaries are the line extensions of $e$ and $f$.
Then we can compute a
piecewise linear rational function $R$
which
maps
a value $t\in[0,1]$
to
the maximum value $r\in[0,1]$,
such that there exists a point $y$ along $l$
for which $\cone{e_t,y,f_r}$
does not contain $x_1$ or $x_2$.
Some values in $[0,1]$
may not have any such valid values,
and we can also compute the ranges of values which are valid.
\end{lemma}

\begin{proof}
First, note that it suffices to find the piecewise linear rational function with respect to only
the \emph{line} passing through $f$, rather than $f$ itself.
Such a function can be combined with two constants for the endpoints via a min and max,
using \cref{lemma:lr-max-min}.

Given the guard point as a parameter along $l$, we can compute the final
value along the line passing through $f$
as the minimum (parametrized) point intersected by the rays from the guard points through $x_1$ and $x_2$ (or $\infty$ if neither intersects),
and this minimum can be taken as a piecewise linear rational function due to \cref{lemma:lr-max-min}.
Since linear rational functions can be composed,
it suffices to find the guard point, or to compute a small number of potential guard points, and to take the maximum over those.

Consider the lines passing through the pairs of points $(e_t,x_1)$, $(e_t,x_2)$, $(x_1,x_2)$,
as well as $e_t$ with each endpoint of $l$.
Consider all the intersections with each of these lines with the line segment $f$ (at least two of which are just the endpoints of $l$ themselves).
We claim that one of these intersection points forms an optimal guard point
(a guard maximizing $r$ for a fixed $t$).
Note that the optimal guard point may not be unique.
This would be enough to complete the proof.
Suppose
for a pair of points $e_t$ and $f_r$
that there is an optimal guard point $y$ along $l$.
Then $\cone{e_t,y,f_r}$ does not contain either of $x_1$ and $x_2$.

We will break into a series of cases.
We have a subfigure for each case in \cref{fig:simple-art-cases}, with labels in parentheses.
If $y$ is an endpoint of $l$ (case a), we are done.
Consider moving $y$ along each direction of $l$ to obtain
maximal-distance points $y'$ and $y''$ along $l$ from $y$
such that
$\cone{e_t,y',f_r}$
and
$\cone{e_t,y'',f_r}$
do not contain either $x_1$ or $x_2$.
If either $y'$ or $y''$ is equal to an endpoint of $l$ (case b),
we are done.
If either pair of points $(e_t,y')$ or $(e_t,y'')$ induce lines containing
$x_1$ or $x_2$
on their boundaries (case c), we are also done.
If the lines induced by either of $(y',f_r)$ or $(y'',f_r)$
contain both $x_1$ or $x_2$ on their boundaries (case d),
we are done.
The only remaining case is that one of $x_1$ or $x_2$
lies on the line induced by $(y',f_r)$ and the other
lies on the line induced by $(y'',f_r)$.
If this happens,
then either each of $x_1$ and $x_2$
lie on opposite sides
of the line through $(f_r,y)$
(case e, a contradiction since then one would be inside a cone for some $y'''$ in between $y'$ and $y''$)
or each of $x_1$ and $x_2$ lie on opposite sides
of $f$
(case f, a contradiction of the half-plane assumption).
\end{proof}

\begin{figure}[ht]
\centering
\begin{minipage}[b]{0.245\textwidth}
    \centering
    \includegraphics[scale=0.40,page=2]{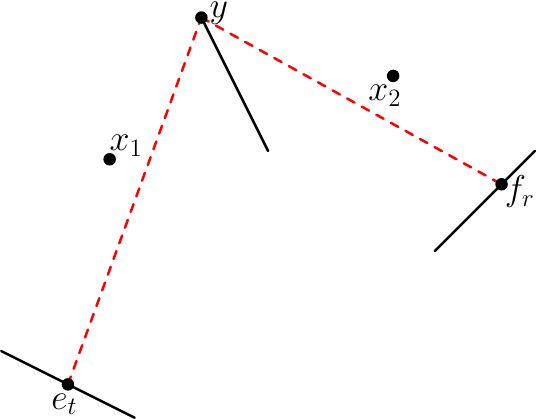}
    \captionsetup{justification=centering}
    \subcaption{}
    \label{fig:case-a}
\end{minipage}
\begin{minipage}[b]{0.245\textwidth}
    \centering
    \includegraphics[scale=0.40,page=4]{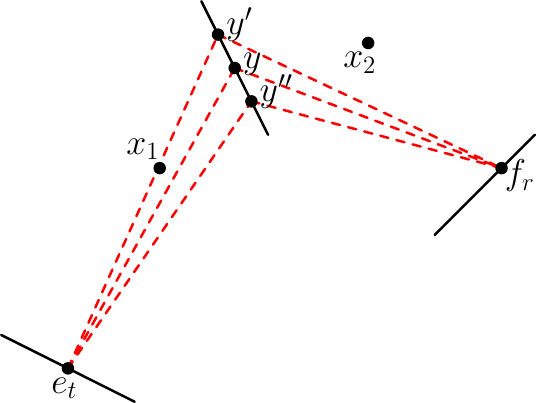}
    \captionsetup{justification=centering}
    \subcaption{}
    \label{fig:case-b}
\end{minipage}
\begin{minipage}[b]{0.245\textwidth}
    \centering
    \includegraphics[scale=0.40,page=5]{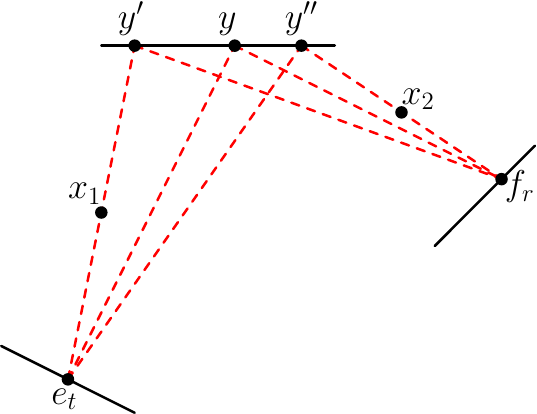}
    \captionsetup{justification=centering}
    \subcaption{}
    \label{fig:case-c}
\end{minipage}%
\begin{minipage}[b]{0.245\textwidth}
    \centering
    \includegraphics[scale=0.40,page=6]{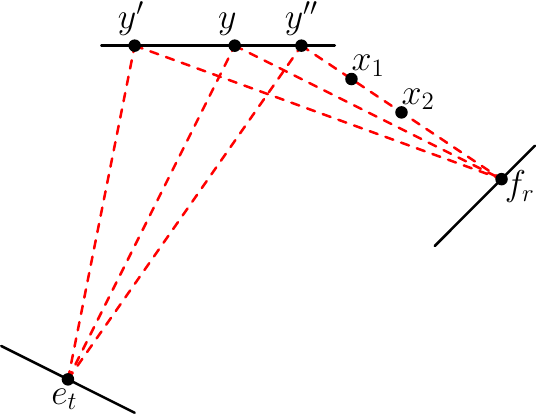}
    \captionsetup{justification=centering}
    \subcaption{}
    \label{fig:case-d}
\end{minipage}

\vspace{0.5em}
\begin{minipage}[b]{0.6\textwidth}
    \centering
    \includegraphics[scale=0.40,page=7]{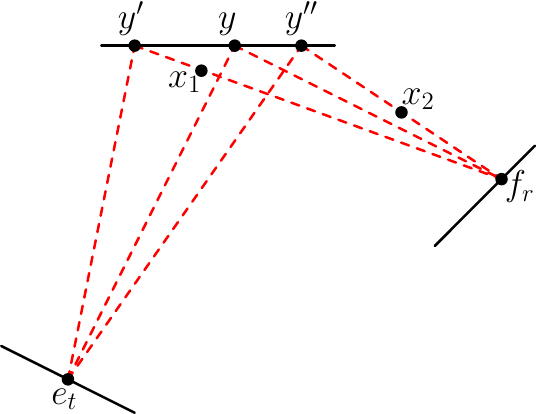}
    \includegraphics[scale=0.40,page=8]{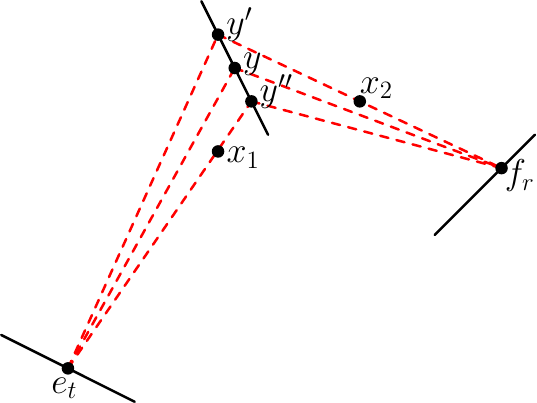}
    \captionsetup{justification=centering}
    \subcaption{ --- two examples}
    \label{fig:case-e1}
\end{minipage}
\begin{minipage}[b]{0.26\textwidth}
    \centering
    \includegraphics[scale=0.40,page=9]{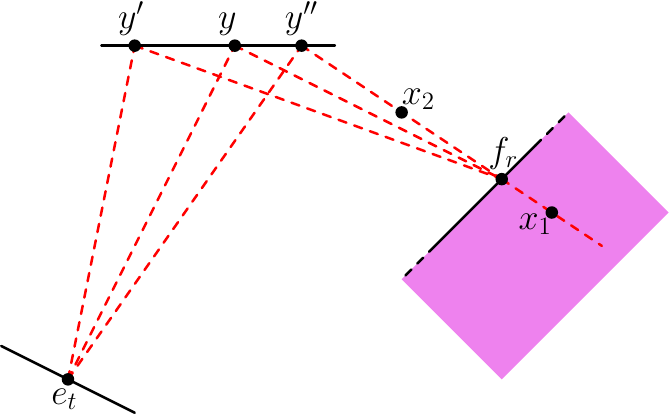}
    \captionsetup{justification=centering}
    \subcaption{}
    \label{fig:case-f}
\end{minipage}
\caption{Visualizations of the cases for \cref{lemma:restricted-toy}.}
\label{fig:simple-art-cases}
\end{figure}

\begin{lemma}
\label{lemma:restricted-toy-2}
Let $e,f$ and $l$ be non-crossing line segments in $\R^2$,
with $e,f$ as before.
Let $X$ be a finite set of points.
Assume $l$ and all points $x \in X$
lie within the intersection of two half-planes
whose boundaries are the line extensions of $e$ and $f$.
Then we can compute
a piecewise linear rational function $R$
which maps
a value $t\in[0,1]$ to
the maximum value $r\in[0,1]$,
such that there exists some point $y$ along $l$
for which $\cone{e_t,y,f_r}$
does not contain any $x\in X$.
Some values in $[0,1]$
may not have any such valid values,
and we can also compute the ranges of values which are valid.
\end{lemma}

\begin{proof}
Take the minimum over applications of \cref{lemma:restricted-toy}
for each (ordered) pair $x_1,x_2$ from $X$,
and apply \cref{lemma:lr-max-min}.
\end{proof}

\begin{lemma}
\label{lemma:restricted-toy-3}
Let $e,f$ be non-crossing line segments in $\R^2$ as before.
Let $X$ be a finite set of points.
Let $C$ be a simple polygon that does not contain any part of $e$, $f$, or $X$ in its interior.
Assume $C$ and each $x\in X$
are all contained within the intersection of a pair of half-planes whose boundaries pass through $e$ and $f$.
Then we can compute a piecewise linear rational function
$R$ which maps a value $t\in[0,1]$ to the maximum value $r\in[0,1]$,
such that there exists some point $y$ in $C$
for which $\cone{e_t,y,f_r}$ does not contain any $x\in X$.
Some values in $[0,1]$
may not have any such valid values,
and we can also compute the ranges of values which are valid.
\end{lemma}

\begin{proof}
Apply \cref{lemma:restricted-toy-2} to each edge of $C$ as $l$,
and apply \cref{lemma:lr-max-min}.
It suffices to show that a guard location on an edge of $C$ suffices
(as opposed to in the interior of $C$).

Fix points $e_t$ and $f_r$ (so that $f_r$ is the optimal choice with respect to $e_t$).
Let $v$ be a valid guard point in $C$
attaining $f_r$,
so that $\cone{e_t,v,f_r}$ does not contain any element of $X$ in its interior.
Since $\cone{e_t,v,f_r}$ intersects $C$,
it must contain within it at least one point $y$ along the boundary of $C$.
If $\cone{e_t,y,f_r}$ covered
any points in $X$ not covered by $\cone{e_t,v,f_r}$,
then the half-plane constraint would be violated, so we are done.
See \cref{fig:visibility-polygon-boundary} for a visualization of this argument.
\end{proof}

\begin{figure}[ht]
\centering
\includegraphics[scale=0.6,page=1]{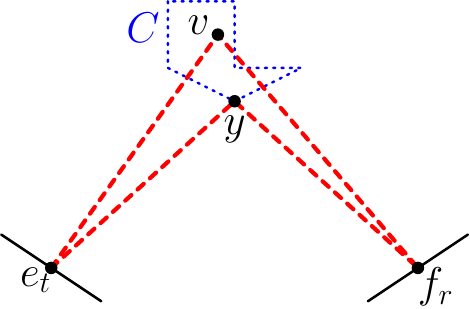}
\caption{A visualization of the argument for \cref{lemma:restricted-toy-3}.}
\label{fig:visibility-polygon-boundary}
\end{figure}

Note that in all of the above series of lemmas, the final range $[a,b]$ may in fact be empty. 
Now we are ready to present the final theorem of this section. A visualization of the lemmas applied here can be found in \cref{fig:art-gallery-fn}.
\begin{theorem}
\label{thm:art-gallery-fn}
Let $P$ be a simple polygon for an instance of \CAG,
with $n=|P|$ and $b$ bits per coordinate.
Then the instance can be reduced to
the (computational) Analytic Arc Cover Problem by way of a piecewise linear rational function $g$,
with polynomially many pieces,
each with $\text{poly}(n,b)$ bits,
and $\text{poly}(n,b)$ bits per endpoint.
\end{theorem}

\begin{proof}
Let the edges of the polygon be $e_0,e_1,\dots,e_{n-1}$.
We can create a mapping of a counterclockwise traversal of these
edges onto $[0,1)$
which itself induces a (circular) traversal of $S^1$ when $[0,1)$ is mapped onto it.
The choice of mapping is unimportant,
e.g. map the traversal of $e_i$ onto the interval $\left[\frac in,\frac{i+1}n\right)$.
Let $g:\R\to\R$ be the next-generator function for this traversal.
That is, $g$ takes a value $t\in[0,1)$,
proof, we will assume $e_i$ is fixed,
since the function can be defined piecewise for each $e_i$.

For each segment $e_j$ we can first check if
there is any part of $e_i$ that $g$ maps to
$e_j$
with \cref{lemma:greedy-choice}.
If so:
Let $X$ be the set of vertices of $P$
along the edges $e_{(j+1)\pmod n},%
\dots,e_{(i-1)\pmod n}$
which are found in the cone formed half-planes whose boundaries overlap $e_i$ and $e_j$
directed into the polygon.
Let $C'$ be the visibility polygon(-al region) that can ``see''
\emph{all} the edges
$e_{(i+1)\pmod n},%
\dots,e_{(j-1)\pmod n}$
(or the shared vertex of $e_i$ and $e_j$ if $j=(i+1)\pmod n$).
Let $C$ be the intersection of $C'$
with the two half-planes whose boundaries pass through each of $e$ and $f$,
with interiors directed into the polygon interior.
Finally, apply \cref{lemma:restricted-toy-3}
on $e=e_i, f=e_j, X=X, C=C$
to obtain a piece-wise linear rational function $f_{i,j}$
defined on some sub-interval of $e_i$.

Let $f_i$ be the function which takes the maximum values
over all applicable $f_{i,j}$ via \cref{lemma:lr-max-min},
and let $f$ be the piecewise function that applies $f_i$ to the $i$th edge.
We claim that, when scaled to the appropriate bounds for $e_i$,
$f$ agrees with
$g$ over $e_i$.
Let $p$ and $q$ be points along $e_i$ and $e_j$ respectively,
so that $g$ maps $p$ to $q$.
Let $y$ be any optimal guard point
(that is, a guard point that
contains the boundary counter-clockwise from $p$ to $q$ 
in its visibility region).
By definition, $y$ must be inside $C$.
Moreover, no point in $X$ can be contained in the cone
formed by $p,y,q$,
since the presence of such a point would violate
the requirement that the lines $\overline{py}$ and $\overline{qy}$
are contained fully within the polygon
(i.e., that the pairs $p,y$ and $q,y$ are each mutually visible).
Hence, $f(t)\geq g(t)$.

Finally, suppose for contradiction
that the guard location $y'$ produced for $f(t)$
did not contain the entire path
from $p$ to $q$ in
its visibility polygon.
Since $y'$ is inside $C$,
\cref{lemma:suffix-prefix-unnecessary}
gives that $y'$ must not contain one of $p$ or $q$
within its visibility polygon.
Without loss of generality assume that it is $p$.
Then the line segment from $y'$ to $p$
must intersect the exterior of $P$,
and hence must intersect some edge $e_k$ of $P$.
We obtain two cases:
\begin{itemize}
    \item If $e_k$ is in $e_{(i+1)\pmod n},\dots,e_{(j-1)\pmod n}$, then
        assume without loss of generality that $(k-i)\pmod n$ is minimal.
        Since $y'$ is in $C$,
        $y'$ is in the visibility polygon of $e_k$,
        and hence $e_k$ has a particular orientation at the crossing point with $\overline{py'}$.
        The minimality of $e_k$ implies that
        $p$ has the opposite orientation.
        Hence $y'$ is not in the half-plane induced by $e_i$,
        and hence not in $C$,
        a contradiction.
    \item If $e_k$ is in $e_{(j+1)\pmod n},\dots,e_{(i-1)\pmod n}$, then
        since no point in $X$ is in
        $\cone{p,y',q}$,
        $e_k$ must cross \emph{both} $\overline{py'}$ and $\overline{qy'}$.
        Hence, $y'$ does not contain \emph{any} of $e_{(i+1)\pmod n},\dots,e_{(j-1)\pmod n}$
        (nor the shared endpoint of $e_i$ and $e_j$ if $j=(i+1)\pmod n$)
        within its visibility region. Thus, $y'$ is not in $C$,
        a contradiction.
\end{itemize}
See \cref{fig:contradiction-cases} for visualizations of each argument.
With this, we have completed the proof.
\end{proof}

\begin{figure}[ht]
\centering
\includegraphics[scale=0.6,page=2]{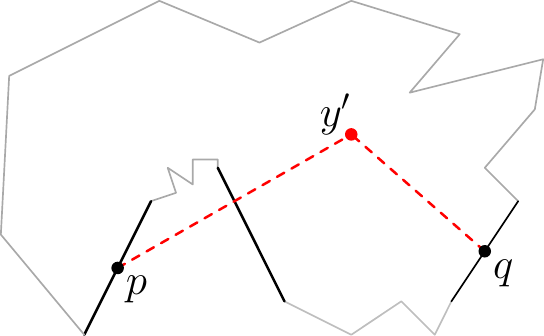}
\hspace{2em}
\includegraphics[scale=0.45,page=3]{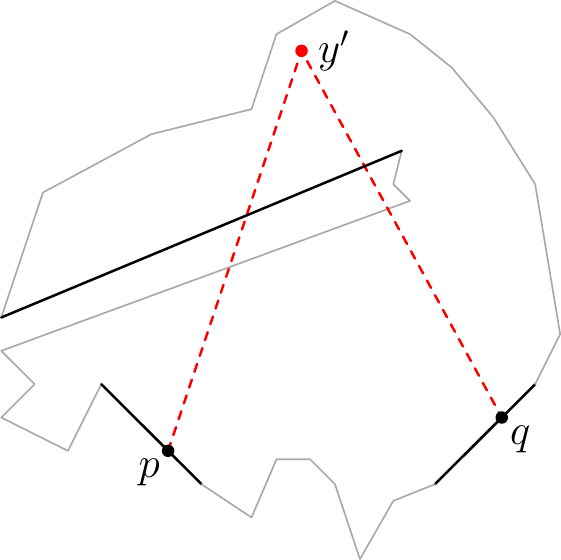}
\caption{Visualizations of the cases used for the final step in the proof for \cref{thm:art-gallery-fn}.}
\label{fig:contradiction-cases}
\end{figure}

\begin{figure}[ht]
\centering
\includegraphics[scale=0.8,page=2]{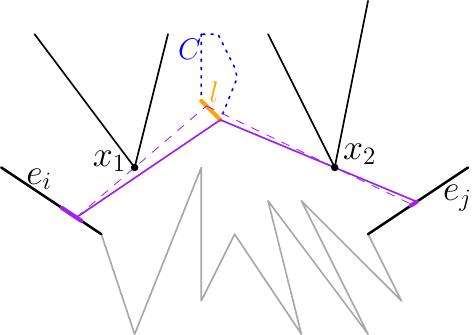}
\caption{A visualization of the construction for the piecewise rational function used for \CAG.}
\label{fig:art-gallery-fn}
\end{figure}

As a consequence of this proof, we obtain \cref{thm:P-cag} by applying \cref{thm:composition_lrf}.

\begin{remark}
\label{rmk:holes}
In the concurrent work by Merrild, Rysgaard, Schou, and Svennin~\cite{MerrildRRS24},
it was also proposed that the case of Polygons with holes may be interesting.
That is, there are polygonal obstacles within the polygon blocking the sightlines of the guards.
Our approach extends to this case without much additional work:
\begin{itemize}
    \item \cref{lemma:suffix-prefix-unnecessary} no longer applies, so its uses have to be updated.
    \item \cref{lemma:greedy-choice} needs a different argument not using \cref{lemma:suffix-prefix-unnecessary}, but known results exist for this that also respect holes~\cite{SuriO86}.
    \item $X$ must include the vertices along these holes.
    \item The proof by contradiction step at the end of the proof of the above theorem needs to be modified. In particular, the WLOG choice of $p$ must be replaced by some interior point of $e_i$ or $e_j$ between $p$ and $q$ along the boundary.
    The second case also needs to also consider the edges along the holes.
\end{itemize}
We can also easily handle the case of covering one of these ``holes'' (rather than the outer boundary) via the same argument.
\end{remark}

\section{The Polygon Separating Problem}
\label{sec:sep}

In this section, we will study \SegS. 
Recall that \SegS considers separating two sets of line segments with a minimum-vertex convex polygon.
We can also consider the problem \SegPolyS, where we wish to separate a convex polygon from a collection of line segments with a minimum-vertex polygon (not necessarily convex).
In fact, these two problems are equivalent.
First, note that a convex polygon contains a set of line segments
if and only if it contains their convex hull,
so \SegS is equivalent to separating a set of line segments
from an ``inner'' convex polygon.
Moreover, it is also easy to check feasibility of a problem instance as a result
--- one only needs to verify that the convex hull of the ``inner'' line segments
does not contain any part of any of the ``outer'' line segments in its interior.
We can make another observation for \SegPolyS:

\begin{theorem}
Let $P$ be a convex polygon
contained by a polygon $Q$.
If $Q$ is non-convex,
there is another convex polygon $Q'$
with at most as many sides as $Q$
such that $P\subset Q'\subset Q$.
\end{theorem}

See \cref{fig:convex-nonconvex-equivalence} for an example.
This theorem is analogous to a similar theorem proved by~\cite{YapABO89},
although we require a different (but still elementary) proof.

\begin{figure}[ht]
\centering
\includegraphics[scale=0.55,page=3]{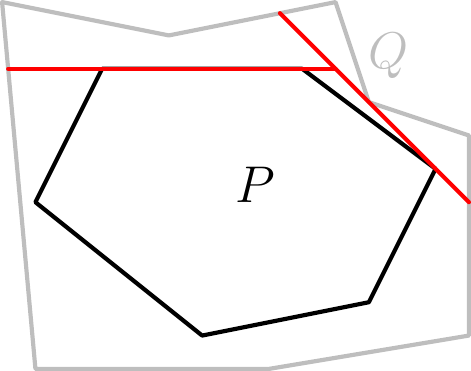}
\hspace{2em}
\includegraphics[scale=0.55,page=4]{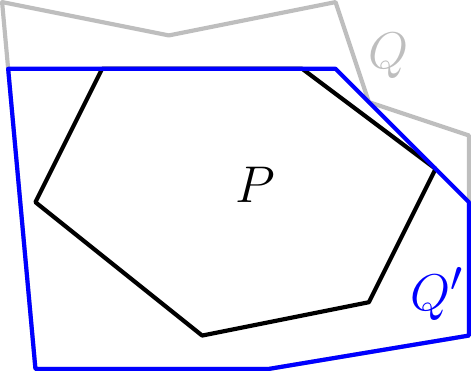}
\caption{A construction for the equivalence between convex and non-convex separating polygons.}
\label{fig:convex-nonconvex-equivalence}
\end{figure}

\begin{proof}
Consider a reflex vertex $v$ of $Q$.
Take some tangent line $L$ of $P$
separating $v$
from the interior of $P$.
$Q$ and $L$ together
separate $\R^2$ into cells,
one of which contains $P$.
The boundary of that cell is then a smaller polygon containing
$P$
with at most as many sides as $Q$,
at least one fewer vertices.
Repeat until a polygon $Q'$ with no reflex vertices is found.
\end{proof}

This gives us an immediate corollary:

\begin{corollary}
\SegS is in \P~if and only if \SegPolyS~is in \P.
\end{corollary}

This is because both problems are equivalent to the problem of separating a set of ``outer'' line segments
from an ``inner'' convex polygon with a separating convex polygon.
We will henceforth refer to the resulting problem \SegCPS.
Note that without loss of generality, it can be assumed that each side of the convex polygon
must touch the boundary of the inner polygon (follows from a standard separating line argument).
See \cref{fig:carving-example} for an example of a \SegCPS.

\begin{figure}[ht]
\centering
\reflectbox{\includegraphics[scale=0.55,page=9]{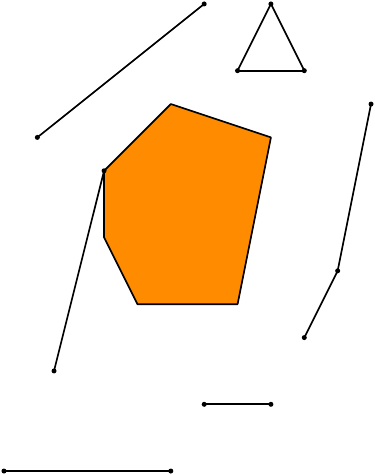}}
\hspace{2em}
\reflectbox{\includegraphics[scale=0.55,page=10]{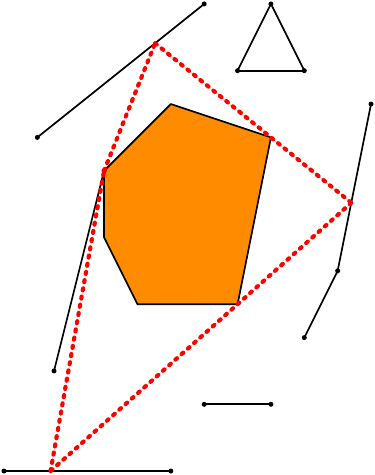}}
\hspace{2em}
\reflectbox{\includegraphics[scale=0.55,page=11]{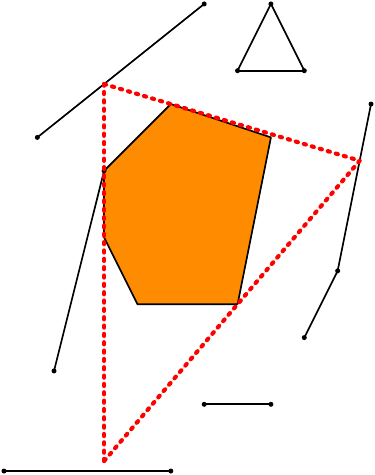}}
\caption{An instance of \SegCPS, and two potential solutions (one sub-optimal).%
}
\label{fig:carving-example}
\end{figure}

If we knew one of the sides of the optimal solution, we could go counterclockwise around the inner polygon and greedily take the next edge that covers as many uncovered segments as possible. This approach would work for \PointS as it can be shown that there exists an edge that passes through two input points. However, for \PS and \SegS, we are unable to prove that this is the case. Indeed, the algorithm of \cite{YapABO89} for \PS also accounts for this possibility.

\begin{theorem}
Consider an instance of the minimum carving by line segments problem
with a convex hull of size $m$, $n$ line segments,
and $\leq b$ bits per coordinate.
Then, it can be reduced to \AAC,
whose next-generator
is a
piecewise linear rational function of two dimensions $g$,
whose domain is the unit circle $S^1$ embedded in $\R^2$.
$g$ has polynomially many pieces,
polynomially many bits per linear rational function,
and polynomially many bits per boundary point.
\end{theorem}

Unlike the art gallery problem we discussed,
it is not as intuitive how this problem can be reduced to the arc cover problem.
The key idea here is that we will map cuts
to points on $S^1$
(through the ray representation in this case).
Specifically, the normal vector of each cut (where the cut is treated as a half-plane away from the inner polygon)
can be mapped to a point on the unit circle with the standard embedding in $\R^2$.
\begin{figure}[ht]
\centering
\reflectbox{\includegraphics[scale=0.55,page=5]{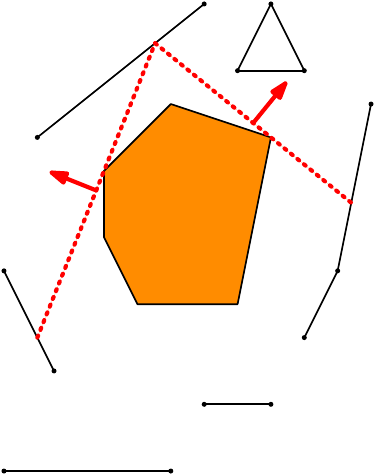}}
\hspace{3em}
\reflectbox{\includegraphics[scale=0.55,page=6]{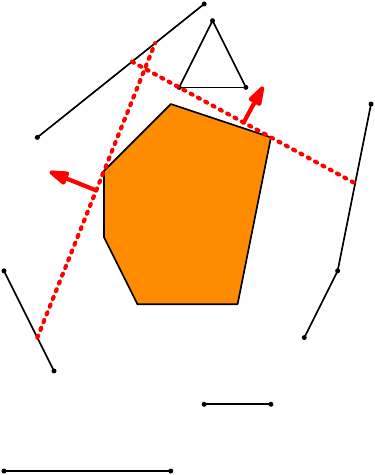}}
\caption{The two types of ``sequential'' cuts, along with their normal vectors.}
\label{fig:carving-example-directions}
\end{figure}

\begin{figure}[ht]
\centering
\reflectbox{\includegraphics[scale=0.55,page=1]{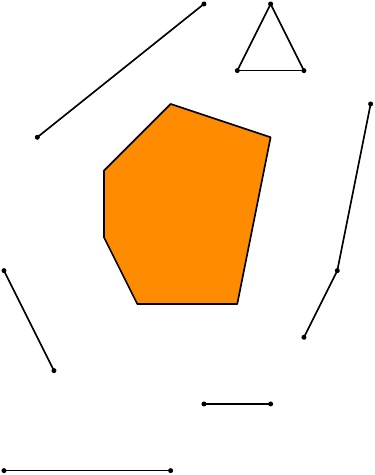}}
\hspace{2em}
\reflectbox{\includegraphics[scale=0.55,page=3]{min-carving-off-line-segments}}
\hspace{2em}
\reflectbox{\includegraphics[scale=0.55,page=4]{min-carving-off-line-segments}}
\caption{An instance of \SegCPS, and two greedy cut sequences from different starting points $x$.
These are the same resulting polygons as \cref{fig:carving-example}.%
}
\label{fig:carving-example-2}
\end{figure}

\begin{proof}
For any cut $C$, sweep counter-clockwise starting from $C$,
with these events:
\begin{enumerate}
    \item An uncovered segment or partial segment begins intersecting the half-plane.
    \item An uncovered segment or partial segment is about to stop intersecting the half-plane.
\end{enumerate}
Terminate the sweep on the first instance of the second kind.
The resulting ``next'' cut takes on two possible forms (
see \cref{fig:carving-example-directions} for an example of each):
\begin{itemize}
    \item The other tangent line from a point along the first line segment intersected by $C$.
        We will refer to this as a greedy reflection.
    \item A tangent line through the endpoint of some other line segment.
\end{itemize}

Our reduction is as follows:
Let $S$ be the set of all lines passing through both a line segment endpoint, $p$ or $q$ for
every line segement $pq$, and another vertex $r$ of the inner polygon.
Let $T$ be the set of tangents into the half-planes induced by each line in $S$
as well as the set of lines $L_{P}$ extending the polygon.
The set $T$ forms the boundaries of our ``outer-most'' pieces, and the next generator at these
boundaries can be computed directly.

Next, within the open interval inside each of these pieces,
there is a unique vertex of the polygon through which we consider the crossing half-planes.
That is, there is one current vertex around which a line is being swept.
Moreover, the set of intersected line segments remains constant
as one sweeps around this vertex within the piece (by definition).
We need only specify the next-generator $g$ within each piece.
Rather than carefully handle the cases in \cref{fig:carving-example-directions},
we take the first of all potential next half-planes
required to fully cover each line segment not covered by the current half-plane.
This can be performed with a minimum operation after taking an appropriate branch cut,
and the resulting function is itself a piecewise linear rational function by \cref{lemma:lr-max-min}.

In particular, there are two types of valid choices of next half-plane:
Either the next half-plane intersects a segment at the same location (but from ``behind''),
or it starts from an endpoint of a fully uncovered line segment.
In either case, it suffices to find the vertex of the inner polygon $P$
for which the corresponding tangent line is formed,
since given that point we can compute
a 2D linear rational function giving the (not-necessarily-unit) normal vector of the next half-plane
subject to the constraints of the ray representation
(see \cref{fig:next-half-plane-formula} for this derivation).
For the endpoints, the next tangent vertex can be computed ahead of time.
For the segments, the next tangent vertex choice
can be also computed as a mapping of partitions of the segment.
Each of these take polynomial time to compute.

We have described the entire linear rational function.
Moreover, the reduction itself takes polynomial time,
so we are done.
\end{proof}

\begin{figure}[ht]
\centering
\includegraphics[scale=0.6]{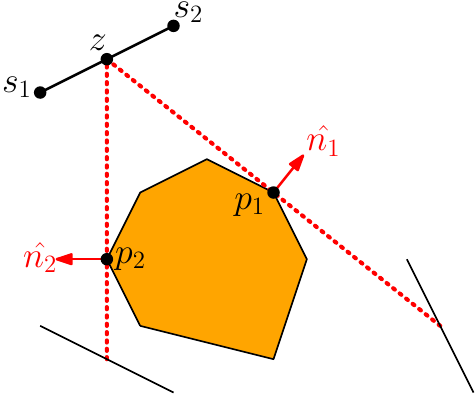}
\caption{A linear rational function for the next half-plane by its normal vector $\hat{n_2}$ from $\hat{n_1}$, with $p_1,p_2,s_1,s_2$ all fixed.
In particular, $z=s_1+\frac{\hat{n_1}\cdot(p_1-s_1)}{\hat{n_1}\cdot(s_2-s_1)}(s_2-s_1)$
and $\hat{n_2}=\left((p_2-z)_y,-(p_2-z)_x\right)$.%
}
\label{fig:next-half-plane-formula}
\end{figure}

By \cref{thm:composition_lrf}, we
have proven \cref{thm:P-seg-s} and have also shown \SegPolyS and \SegCPS to be in \P.

\section{Optimal 3D Half-Plane Carving}
\label{sec:carving}

Recent work by Robson, Spalding-Jamieson, and Zheng~\cite{RobsonSJZ2024} has presented
an algorithm for determining if a 3D polytope can be ``carved''
with half-plane cuts.
That is,
the algorithm decides if there exists a set $H$ of half-planes
such that the polytope $P$ is a connected component of
$\R^3\setminus\left(\bigcup_{H\in\mathcal H}H\right)$.
Examples of ``YES'' and ``NO'' instances for this problem can be found in
\cref{fig:3d-carving-example-yes} and \cref{fig:3d-carving-example-no}.
This solved an open problem posed by Demaine, Demaine, and Kaplan~\cite{DemaineDK01},
who considered a 2D variant of this problem
utilizing line segment cuts.
The latter work also posed another open problem in their 2D model:
Does there exist an algorithm to minimize the number of cuts.
In this section, we will solve this problem for the 3D model with half-plane cuts.
That is, we will compute a minimum-sized set of half-planes $\mathcal H$ carving $P$.

\begin{figure}[ht]
\centering
\includegraphics[scale=0.15,page=1]{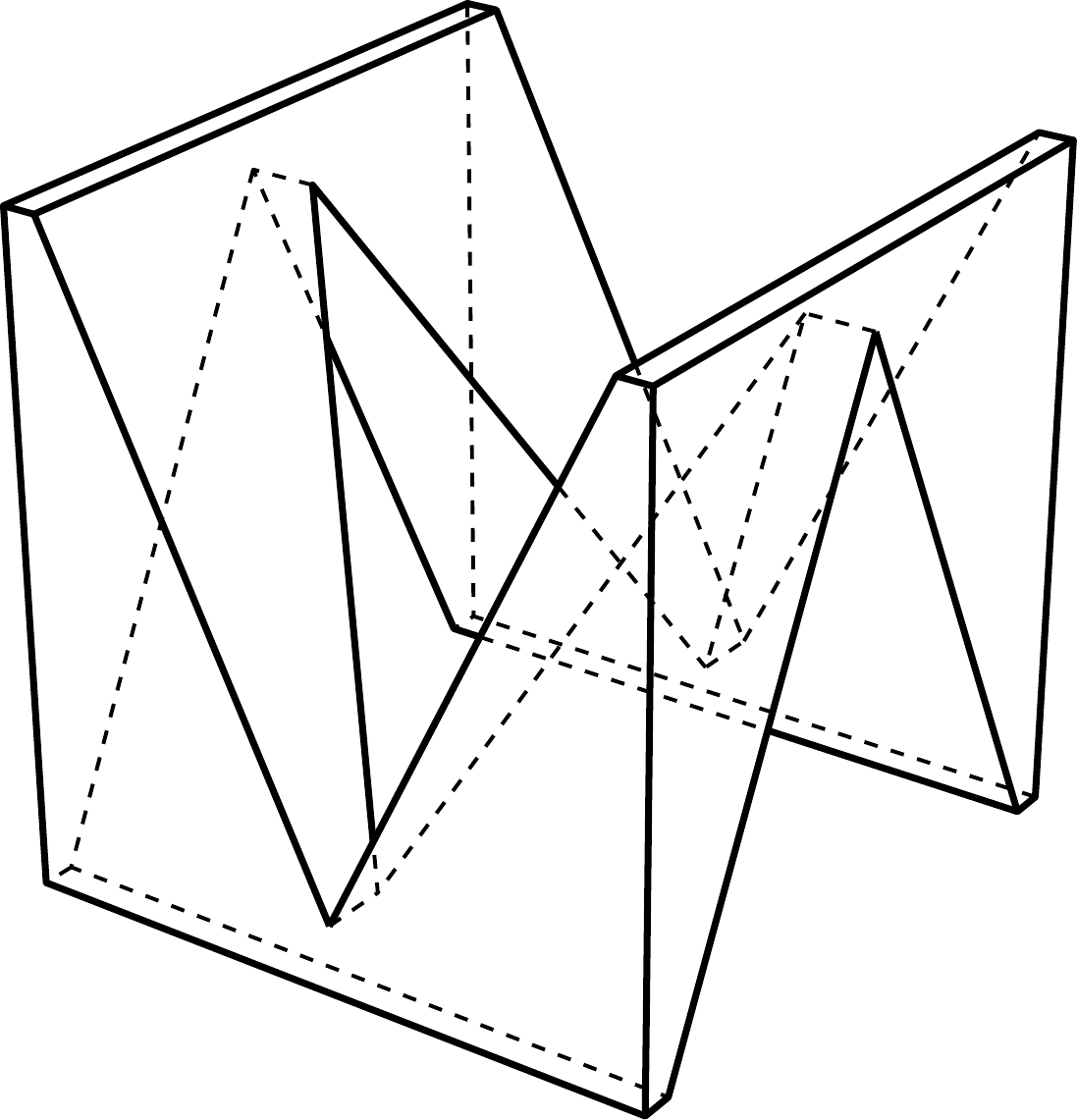}
\hspace{1em}
\includegraphics[scale=0.15,page=1]{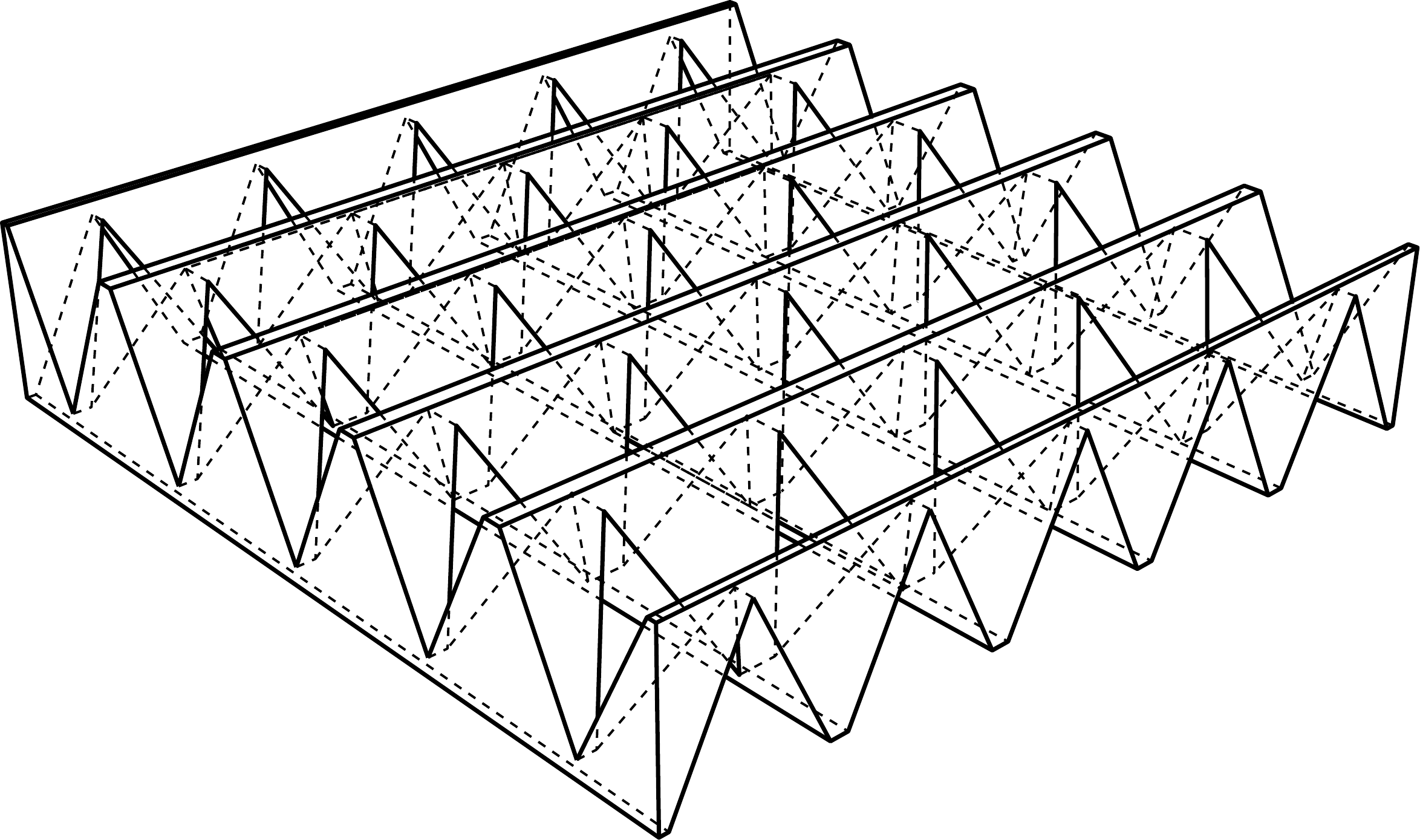}
\hspace{1em}
\includegraphics[scale=0.15,page=1]{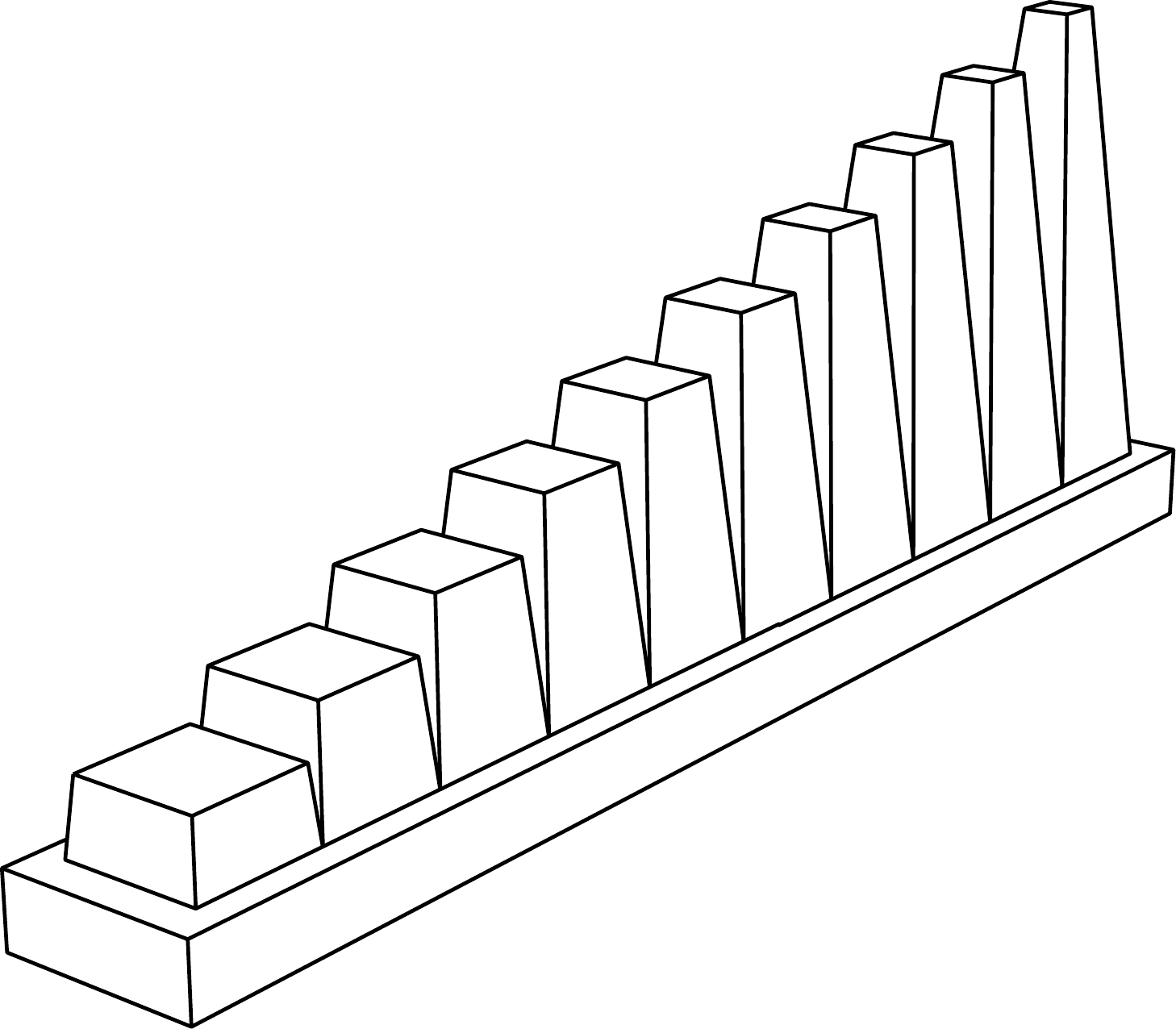}
\caption{Some examples of 3D polytopes that \emph{can} be carved with half-planes.}
\label{fig:3d-carving-example-yes}
\end{figure}

\begin{figure}[ht]
\centering
\includegraphics[scale=0.15,page=1]{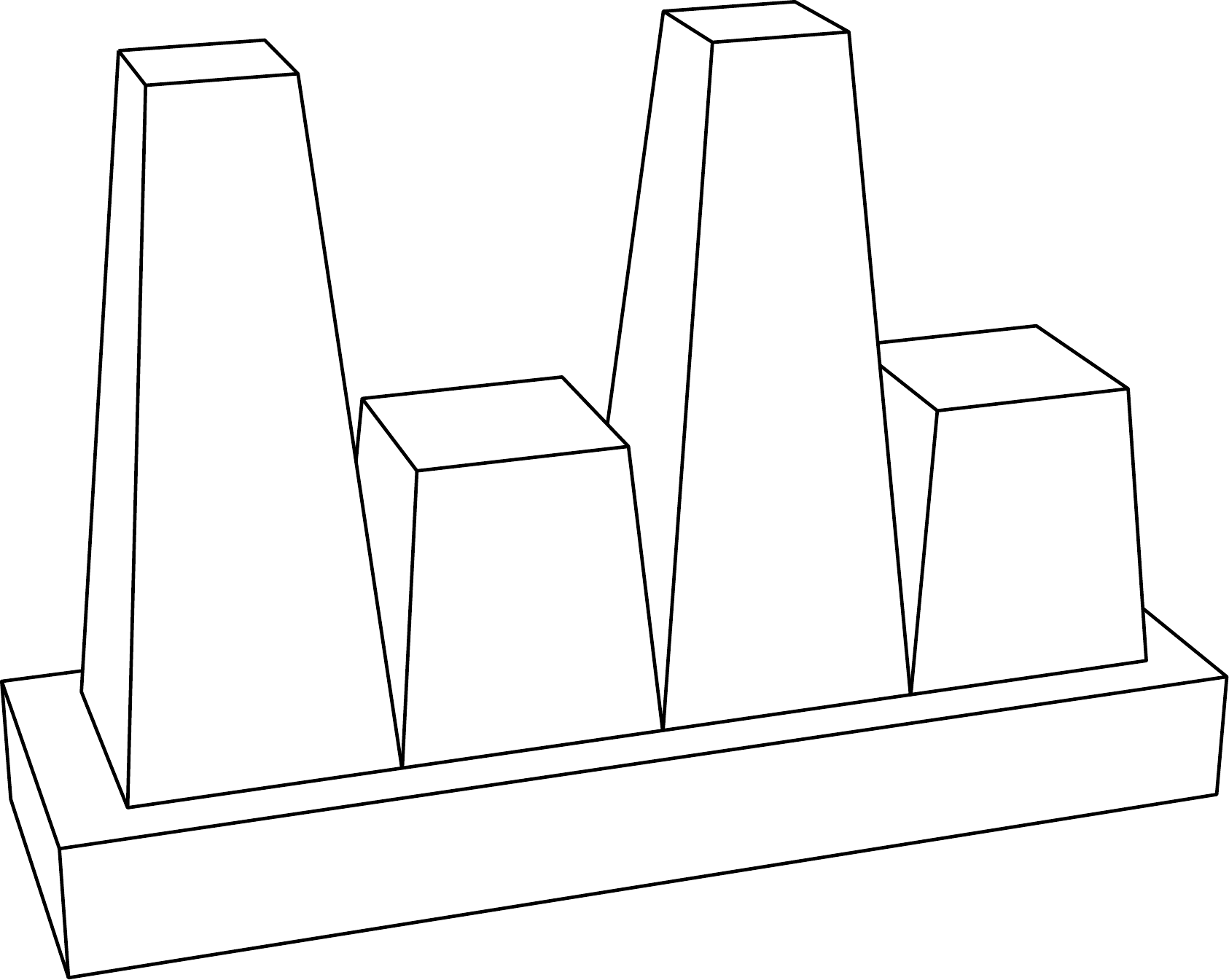}
\includegraphics[scale=0.15,page=1]{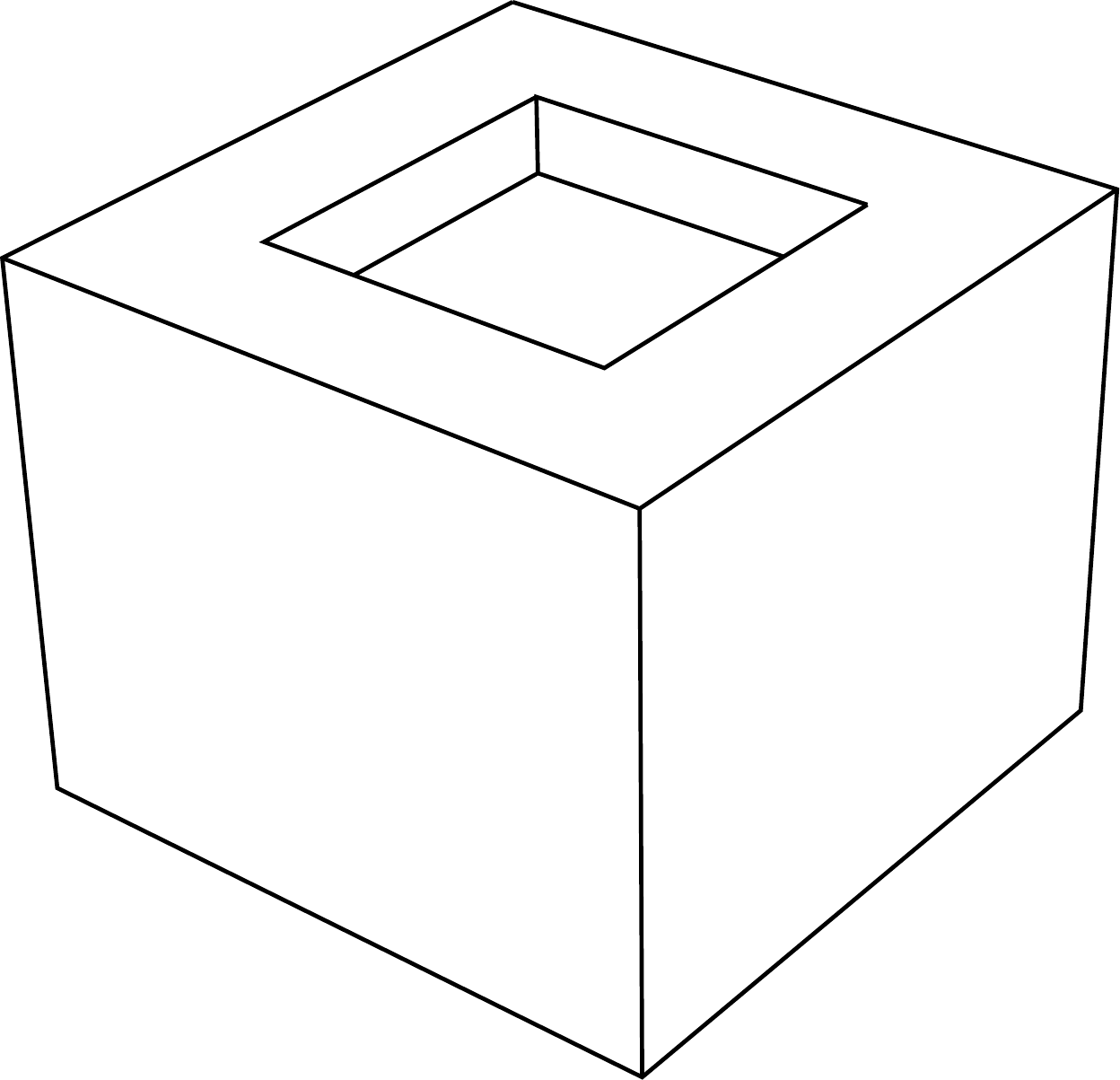}
\caption{Some examples of 3D polytopes that \emph{cannot} be carved with half-planes.}
\label{fig:3d-carving-example-no}
\end{figure}

Importantly, Robson, Spalding-Jamieson, and Zheng proved the following theorem:
\begin{theorem}[{\cite[thm 1]{RobsonSJZ2024}}]
\label{thm:half-plane-cutting-equivalence}
Let $\CH{S}$ denote the convex hull of a set of points $S$.
Let $\Int{P}$ denote the interior of a polytope $P$.
Then, for a triangulated 3D polytope $P$ (that is, each face is triangulated),
the following are equivalent:
\begin{alphaenumerate}
    \item $P$ is carveable with a set of half-planes.
    \item For each facial triangle $T$ in $P$,
        let $L_T$ denote the plane containing $T$.
        There is a single half-plane $H$ containing $T$
        and
        disjoint from
        $\CH{\Int{P} \cap L_T}$
        whose boundary line passes through a vertex of 
        $\CH{\Int{P} \cap L_T}$,
        and also passes through a vertex $v$ of $T$.
\end{alphaenumerate}
\end{theorem}

An example of the latter condition can be found in \cref{fig:3d-carving-separator}.
In particular, a set of all half-planes
corresponds to the latter condition
if and only if it corresponds to
a set of carving half-planes.
For additional details and examples, we point readers to the original paper~\cite{RobsonSJZ2024}.

\begin{figure}[ht]
\centering
\includegraphics[scale=0.15,page=1]{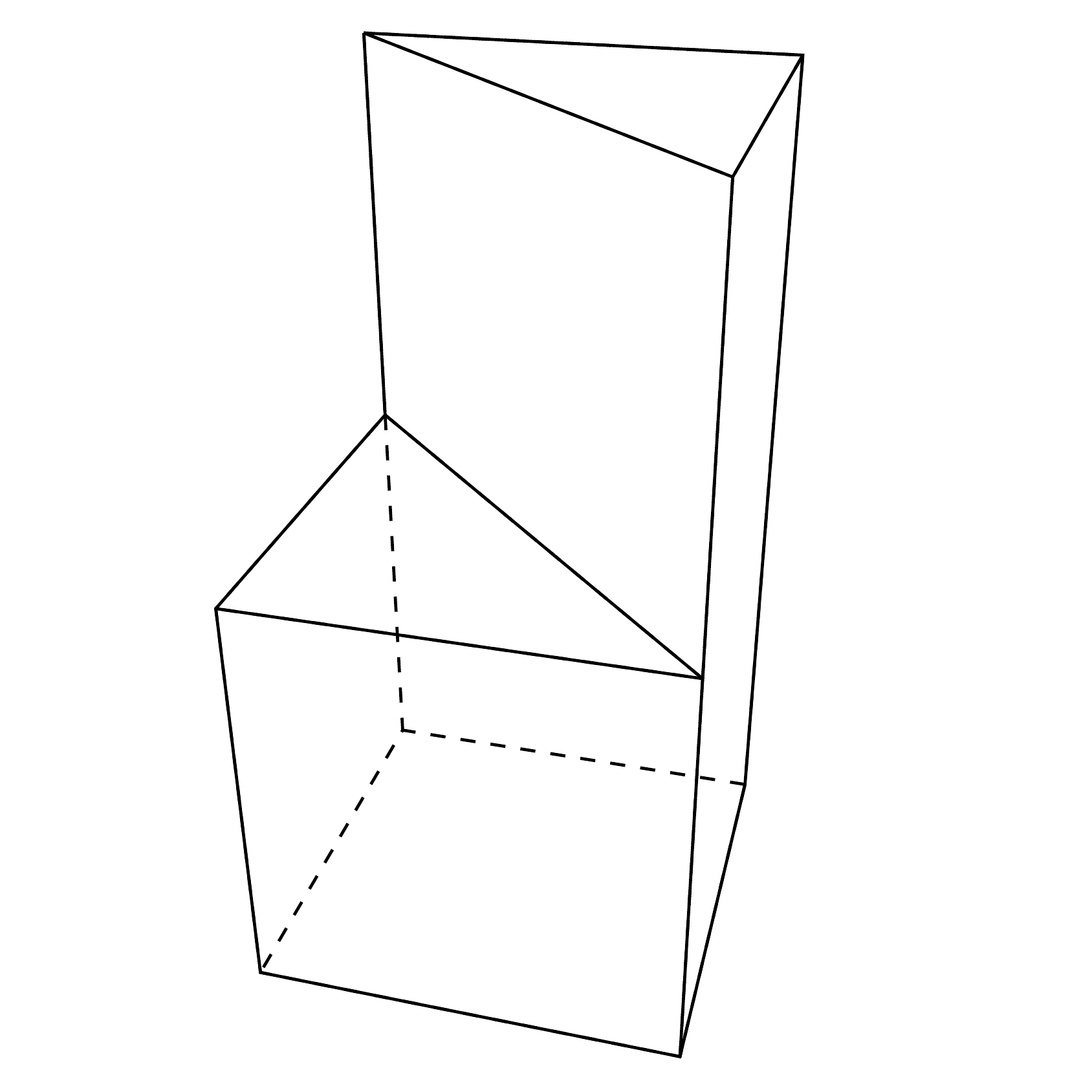}
\hspace{1em}
\includegraphics[scale=0.15,page=4]{saws/half-plane-cut-example}
\caption{An example of a half-plane cut for a particular face.}
\label{fig:3d-carving-separator}
\end{figure}

That said, while such a set of half-planes has size $n$ (for $n$ triangular faces of the polytope),
and can be found in polynomial time via the methods of the original paper,
it does not necessarily result in the fewest number of cuts.
In particular, it may be possible to replace a set of co-planar half-planes with a smaller set.
In fact, this is the only case we must consider: For a face to be carved, there must be a set of co-planar half-planes covering the face.
Hence, this allows us to reduce the optimization variant to the two-dimensional Polygon Separating problem:
\begin{lemma}
Let $P$ be a polytope with triangulated faces
that can be carved with half-planes.
Partition the triangular faces by coplanarity.
For each partition containing a set of triangles $\mathcal T$ along a plane $L_{\mathcal T}$,
let $Q_{\mathcal T}$ be the polygon formed by $\Int{P}\cap L_{\mathcal T}$,
and let $S_{\mathcal T}$ be the set of edges of triangles in $\mathcal T$.
Solve the Polygon Separating problem
on each $Q_{\mathcal T},S_{\mathcal T}$.
Then, the resulting set of half-planes also forms a minimum set of half-planes carving $P$.
\end{lemma}

This is enough to obtain a polynomial-time solution to the optimization problem.
That is, applying \cref{thm:P-seg-s} and \cref{thm:composition_lrf} proves \cref{thm:P-3d}.

\bibliography{citations}

\end{document}